\documentclass [12pt]{article}

\usepackage{amsmath,amsthm,amscd,amssymb}

\usepackage[small, centerlast, it]{caption}
\usepackage{epsfig}
\usepackage[titles]{tocloft}
\usepackage[latin1]{inputenc}
\usepackage{verbatim} %for comments via \begin{comment}
\usepackage[active]{srcltx} %inverse search
\usepackage{fancyhdr}
\usepackage{caption}

\def\sp{\hskip -5pt}

\def\b1{{1\!\!1}}

\def\cE{{\ca E}}
\def\cF{{\ca F}}

\def\cI{{\ca I}}
\def\cJ{{\ca J}}
\def\cL{{\ca L}}

\def\cO{{\ca O}}

\def\cS{{\ca S}}

\def\sH{{\mathsf H}}
\def\sK{{\mathsf K}}
\def\sP{{\mathsf P}}

\def\sS{{\mathsf S}}

\def\bC{{\mathbb C}}           %%%  complex numbers and so on

\def\bI{{\mathbb I}}

\def\bR{{\mathbb R}}

\def\bS{{\mathbb S}}

\def\mL{\mathcal L}

\def\gB{{\mathfrak B}}

\def\gO{{\mathfrak O}}
\def\gP{{\mathfrak P}}

\def\gS{{\mathfrak S}}
\def\gT{{\mathfrak T}}

\def\beq{\begin{eqnarray}}
\def\eeq{\end{eqnarray}}

\newcommand{\ca}[1]{{\cal #1}}         %%  calligraphic

\def\p{\parallel}

\usepackage{sectsty}
\sectionfont{\normalsize}%\normalfont
\subsectionfont{\normalsize\normalfont\itshape}
\newtheoremstyle{thm}
{12pt}% space above
{12pt}% space below
{\itshape}% body font
{}% h indent amount
{\itshape\bfseries}% theorem head font
{}% punctuation after theorem head
{1em}% space after theorem head
{}% theorem head spec (can be left empty, meaning `normal')

\theoremstyle{thm}
\newtheorem{theorem}{Theorem}
\newtheorem{lemma}[theorem]{Lemma}
\newtheorem{proposition}[theorem]{Proposition}
\newtheorem{definition}[theorem]{Definition}

\begin{document}

\hfill{\sl  August  2014} 
\par 
\bigskip 
\par 
\rm

%%%%%%%%%%%%%   Title %%%%%%%%%%%%%%%%%%%%%%%%%%

\par
\bigskip
\large
\noindent
{\bf A geometric Hamiltonian description of composite quantum systems and quantum entanglement}
\bigskip
\par
\rm
\normalsize 

%%%%%%%%%%%%%%%%%%%%%%%%%%%%%%%%%%%%%%%%%%%%%
%%%%%%%%%%%% Authors %%%%%%%%%%%%%%%%%%%%%%%%

\noindent  {\bf Davide Pastorello}\\
\par
\noindent Department of  Mathematics, University of Trento \& INFN (section of Trento), via Sommarive 14, 38123 Povo (Trento), Italy.\\
\par
\noindent   e-mail: pastorello@science.unitn.it
\smallskip

 \normalsize

\par

\rm\normalsize

%\noindent {\small Version of \today}

%\linespread{1.5}
\rm\normalsize

%%%%%%%%%%%% Date %%%%%%%%%%%%%%%%%%%%%%%%%%

\par
\bigskip

\noindent
\small

\section*{Abstract}
Finite-dimensional Quantum Mechanics can be geometrically formulated as a proper classical-like Hamiltonian theory in a projective Hilbert space. The description of composite quantum systems within the geometric Hamiltonian framework is discussed in this paper. 
As summarized in the first part of this work, in the Hamiltonian formulation  the phase space of a quantum system is the K\"ahler manifold given by the complex projective space $\sP(\sH)$ of the Hilbert space $\sH$ of the considered quantum theory. However the phase space of a bipartite system must be $\sP(\sH_1\otimes\sH_2)$ and not simply $\sP(\sH_1)\times\sP(\sH_2)$ as suggested by the analogy with Classical Mechanics. A part of this paper is devoted to manage this problem. 
In the second part of the work, a definition of quantum entanglement and a proposal of entanglement measure are given in terms of a geometrical point of view (a rather studied topic in recent literature). Finally two known separability criteria are implemented in the Hamiltonian formalism. 

\section{Introduction}

%%%%%%%%%%Motivations%%%%%%%%%%%%%%%%%

In recent years there have been several attempts to study quantum entanglement considering the geometrization of finite-dimensional  Quantum Mechanics, exploiting the representation of pure states as points of the complex projective space $\sP(\sH)$ (constructed on the Hilbert space $\sH$ of the theory) as a real manifold with a complex K\"ahler structure. Interesting examples in this field are: \cite{AGMV2} where the authors consider invariant operator-valued tensor fields on the local unitary group $U(n)\times U(n)$ to identify entanglement candidates, while in \cite{AGMV} the authors use the pull-back of Fubini-Study metric to the orbits of the local unitary group. Or else \cite{kus} where
equally entangled states are classified in terms of orbits of (unitarily represented) group actions on the space of states and the degree of the degeneracy of the symplectic form is suggested as an entanglement measure. In \cite{BZ} there is a huge characterization of entanglement and its measure in terms of Fubini-Study metric and induced distance on projective space. 

In this paper a different point of view is adopted to state a proposal of entanglement measure within geometric framework, it is based on the description of quantum states in terms of (Liouville) denisities on projective space/phase space \cite{Gibbons,AS,DV2}. In particular we assume $\dim\sH>2$ in order to characterize states and observables completely in terms of $\cL^2$-functions on $\sP(\sH)$, exploting the results presented in \cite{DV1,DV2}.

A geometrical Hamiltonian notion of entanglement measure is given in terms of the Liouville measure induced by the symplectic form on $\sP(\sH)$. Moreover it is shown how the machinery developed in \cite{DV2} (in particular one-to-one correspondence between operators and frame functions) allows to translate two well-known separability criteria \cite{thirring,Horodecki,Witte} for mixed states in terms of pure geometric Hamiltonian formulation.

One of the crucial points of geometric formulation on projective space is the the passage from Hermitian operators (quantum observables) to real-valued functions. For example, given a Hermitian operator $A$ on $\sH$ one can define the function:
$$e_A(\psi):=\frac{\langle\psi|A\psi\rangle}{\langle\psi|\psi\rangle}\qquad\psi\in\sH\setminus \{0\},$$
which can be viewed as the pull-back of a function defined on $\sP(\sH)$ \cite{AGMV}. The eigenvectors of $A$ are the critical points of $e_A$, i.e. the vectors $\psi_0$ such that $de_A(\psi_0)=0$ and the values of $e_A$ at the critical points are the eigenvalues of $A$ \cite{CCGM}. Thus the operator $A$ can be recovered by $e_A$ thorough its spectral decomposition:
$$A=\sum_{\psi_0\, critical} e_A(\psi_0) P_{\psi_0}\qquad P_{\psi_0}=\frac{|\psi_0\rangle\langle\psi_0|}{\langle\psi_0|\psi_0\rangle}.$$
\\
In this paper we associate the function acting on the projective space $f_A:\sP(\sH)\rightarrow \bR$ to each quantum observable $A$, in the following way:
 $$f_A(p)=(n+1)tr(Ap)-tr(A)\qquad n=\dim\sH,$$
we prove  that the operator $A$ can be recovered by $f_A$ as:
$$A=\int_{\sP(\sH)}f_A(p)\, p\, d\nu(p),$$
where $\nu$ is the Liouville measure induced by the symplectic form on $\sP(\sH)$, this expression has to be interpreted in view of 1-to-1 correspondence of $\sP(\sH)$ with rank-1 projectors on $\sH$, in this sense it is a sort of spectral decomposition.

Let us briefly summarize some notions about our framework: Starting from a celebrated work of Kibble \cite{Kibble} a \emph{quantum Hamiltonian formulation} for finite-dimensional quantum systems has been developed (a celebrated work in this regard is  \cite{AS}). The phase space of the Hamiltonian theory is the projective space $\sP(\sH_n)$, constructed out the $n$-dimensional Hilbert space $\sH_n$, that is a real $(2n-2)$-dimensional smooth manifold with an almost K\"ahler structure. Let us give a quick introduction to this geometric structure. $u(n)$ denotes the Lie algebra of the unitary group $U(n)$, then $iu(n)$ denotes the real vector space of self-adjoint operator on $\sH_n$ and $\gB(\sH_n)$ is the $C^*$-algebra of bounded\footnote{Every linear operator is bounded (and trace-class) on a finite dimensional Hilbert space, but we mantain the general terminology also adopted in infinite dimension.}  operators on $\sH_n$ that is the concrete \emph{observable algebra} in the standard formulation of Quantum Mechanics.

Fixed a point $p\in\sP(\sH_n)$, one can prove that any tangent vector $v\in T_p\sP(\sH_n)$ has the form $v=-i[A_v,p]$ for some $A_v\in iu(n)$. A symplectic form can be defined on $\sP(\sH_n)$ as:
\begin{equation}\label{symp}
\omega_p(u,v):=-i\kappa \,tr([A_u,A_v]p)\qquad u,v\in T_p\sP(\sH_n),
\end{equation}         
with $\kappa>0$. Moreover a Riemann metric can be defined on the manifold:

 \begin{equation}\label{fs}
g_p(u,v):=-\kappa \,tr\left(([A_u,p][A_v,p]+[A_v,p][A_u,p])p\right)\qquad u,v\in T_p\sP(\sH_n),
\end{equation} 

\vspace{0.5cm}
\noindent
called {\bf Fubini-Study metric}. Considering the map given by:

$$j_p:T_p\sP(\sH_n)\ni v\mapsto i[v,p]\in T_p\sP(\sH_n), $$
\\
we have $p\mapsto j_p$ is smooth  and $j_pj_p=-id$ for any $p\in\sP(\sH_n)$. The following relation holds:

$$\omega_p(u,v)=g_p(u,j_pv)\quad \forall p\in\sP(\sH_n).$$
\\
Thus $(\omega,g,j)$ is an almost K\"ahler structure on $\sP(\sH_n)$, in particular there is a symplectic structure where notions of Hamiltonian fields, Poisson brakets, Liouville measure can be defined. 

In \cite{DV2}, we discussed all possible prescriptions to associate each quantum observable $A\in iu(n)$ to a classical observable $f_A:\sP(\sH_n)\rightarrow\bR$, i.e. a real function defined on the projective space $\sP(\sH_n)$ in order to obtain a self-consistent Hamiltonian classical-like theory. Similarly we investigated how to associate each density matrix $\sigma$ to a Liouville density $\rho_\sigma:\sP(\sH_n)\rightarrow\bR$ for computing classical-like expectation values. Imposing several physical requirements \cite{DV2} all the prescriprions to set up a meaningful classical-like formulation are labelled by a positive real number $\kappa$ and given by the so-called \emph{inverse quantization maps}:

\begin{equation}\label{O}
\cO:iu(n)\ni A\mapsto f_A,
\end{equation}
with 
\begin{equation}\label{o}
f_A(p)=\kappa tr(Ap)+\frac{1-\kappa}{n}tr(A)\qquad \kappa>0,
\end{equation}

\vspace{0.5cm}
\noindent
for observables. And about states:
\begin{equation}\label{S}
\cS:\sS(\sH_n)\ni \sigma\mapsto \rho_\sigma,
\end{equation}
with 
\begin{equation}\label{s}
\rho_\sigma(p)=\frac{n(n+1)}{\kappa} \,tr(\sigma p)+\frac{\kappa-(n+1)}{\kappa}\qquad \kappa>0,
\end{equation}

\vspace{0.3cm}
\noindent
where $\sS(\sH_n)$ denotes the set of density matrices on $\sH_n$, i.e. the positive and normalized trace-class operators. The positive constant $\kappa$ is a degree of freedom of the whole theory appearing in definitions (\ref{symp}) and (\ref{fs}). Using the maps $\cO$ and $\cS$ to obtain classical-like observables and states we have this remarkable result for any $\kappa>0$:

\begin{equation}
\langle A\rangle_\sigma=tr(A\sigma)=\int_{\sP(\sH_n)} f_A \rho_\sigma d\nu,
\end{equation}
\\
where $\nu$ is the Liouville measure w.r.t. the symplectic form $\omega$. 

A good choice is $\kappa=n+1$: The classical observable associated to $A\in iu(n)$ and the Liouville density associated to $\sigma\in\sS(\sH_n)$ take respectively the form:

\begin{equation}\label{osst}
f_A(p)=(n+1)tr(Ap)-tr(A)\qquad\mbox{and}\qquad \rho_\sigma(p)=tr(\sigma p),
\end{equation}
\\
the substitution $\kappa=n+1$ in (\ref{s}) produces $\rho_\sigma=n\,\,tr(\sigma p)$, but we can remove the multiplicative factor $n$ changing the normalization of the measure $\nu\rightarrow n\nu$. While in the choice $\kappa=1$ the classical-like observables are the standard expectation values functions (as in \cite{AS} and \cite{Gibbons}) but the states are not positive and normalized then the interpretation of probability densties drops down. For this reason we adopt the choice $\kappa=n+1$.
\vspace{0.5cm}
\section{Observables and states in geometric Hamiltonian formulation}

Every point $p$ of  $\sP(\sH_n)$ can be interpreted as a rank-1 orthogonal projector $p=|\psi\rangle\langle\psi|$ with $\psi\in\sH_n$ and $\parallel\psi\parallel=1$, thus for any pair $p,p'\in\sP(\sH_n)$ the distance $d_2(p,p')=\sqrt{\parallel\psi\parallel^4+\parallel\psi'\parallel^4-2|\langle\psi|\psi'\rangle|^2}$ is well-defined and $d_2(p,p')=\sqrt{2}$ if and only if $\psi\perp\psi'$. A set $N\subset\sP(\sH_n)$ is said to be a \textbf{basis} of  $\sP(\sH_n)$ if  $d_2(p,p')=\sqrt{2}$ for $p,p'\in N$ with $p\not =p'$ and $N$ is a maximal set w.r.t. this property.

\begin{definition}
A map $F:\sP(\sH_n)\rightarrow \bC$ is called  \textbf{frame function} on $\sH_n$ if there is a number $W_F\in\bC$ such that:
$$\sum_{p\in N} F(p)=W_F\quad \mbox{for every basis}\,\,N\,\,\mbox{of}\,\,\sP(\sH_n).$$
\end{definition}

\noindent
Frame functions, introduced by Gleason to prove his celebrated theorem \cite{Gleason}, play a crucial r\^ ole in the characterization of quantum observables and quantum states in terms of functions on the projective space within the classical-like Hamiltonian formulation.

Consider the vector space $\cL^2(\sP(\sH_n),\nu)$ of squared $\nu$-integrable functions on $\sP(\sH_n)$, let us introduce the closed subpace:

\beq
\cF^2(\sH_n):=\left\{F:\sP(\sH_n)\rightarrow\bC| F\in\cL^2(\sP(\sH_n),\nu), F\,\mbox{is a frame function}\right\} \eeq
\\
One can prove that if $n>2$ then $\cF^2(\sH_n)$ is isomorphic to $\gB(\sH_n)$ \cite{DV1,DV2}.  For this reason, we assume $\dim\sH_n=n>2$ without further specifications.

An isomorphism of vector spaces from $\gB(\sH_n)$ to $\cF^2(\sH_n)$ can be defined by the linear extension of $\cO$ (or $\cS$).

The extended action of $\cO$ can be used to define the $C^*$-algebra of observables in terms of frame functions $(\cF^2(\sH_n),\star)$ where $\star$ is the *-algebra product given by:
\begin{equation}
f\star g:=\cO(\cO^{-1}(f)\cO^{-1}(g))\qquad \forall f,g\in\cF^2(\sH_n).
\end{equation}
The explicit calculation for $\kappa=n+1$ yields:
\begin{equation}
  f \star g = \frac{i}{2}\{f, g\} + \frac{1}{2} G(df,dg) + \frac{n}{n+1} \left(\frac{fg}{n}- \int_{\sP(\sH_n)} \sp \sp\sp\sp f g d\nu  - f\sp  \int_{\sP(\sH_n)}\sp \sp\sp\sp g d\nu-g  \sp \int_{\sP(\sH_n)}\sp \sp\sp\sp f d\nu \right), 
\end{equation}
\\
where $G$ is the scalar product on one-forms induced by Fubini-Study metric. The $*$-involution is given by the complex conjugation, and the $C^*$-norm can be defined as $|||f|||:=\p\cO^{-1}(f)\p$ (explicit calculation in \cite{DV2}); thus $\cO$ is an isomorphism of $C^*$-algebras and the physical observables are given by the real functions in $\cF^2(\sH_n)$.

The simplest way to obtain an isomorphism of vector spaces, is taking the extension of the inverse quantization map for states (in the choice $\kappa=n+1$), $\cS:\gB(\sH_n)\rightarrow \cF^2(\sH_n)$:

\begin{equation}\label{S}
\cS(\sigma)(p):=tr(\sigma p)\qquad \forall\sigma\in\gB(\sH_n).
\end{equation}  
\\
The subset in the range of $\cS$ representing the quantum states as \emph{Liouville densities} is given by the image of the set $\sS(\sH_n)$ of density matrtices through $\cS$; Denoting it with the same name:

\begin{equation} \label {frame states}
\sS(\sH_n)=\left\{\rho\in\cF^2(\sH_n)\big |\rho(p)\geq0 \,\forall p\in\sP(\sH_n)\,,\, W_\rho=1\right\}.
\end{equation}
\\By definition, Liouville densities are valued in $[0,1]$, in agreement with interpretation of probability density. Since $\cS:\gB(\sH_n)\rightarrow \cF^2(\sH_n)$ preserves the convex structure of the set of states,  $\sS(\sH_n)$ is a convex set in $\cF^2(\sH_n)$ and its extremal elements represent pure states. To give a complete characterization of states in terms of frame functions, let us invoke the trace-integral formulas \cite{DV2}:

\begin{equation}\label{trace-integral}
\int_{\sP(\sH_n)} \rho(p)d\nu(p)=tr(\cS^{-1}(\rho))\qquad\forall \rho\in\cF^2(\sH_n).
\end{equation}
\begin{equation}\label{trace-integral2}
\int_{\sP(\sH_n)} \overline{\rho(p)}\rho'(p)d\nu(p)=\frac{1}{n+1}(tr(A^\dagger B)+tr(A^\dagger)tr(B)),
\end{equation}
\\
\\
where $A=\cS^{-1}(\rho)$ and $B=\cS^{-1}(\rho')$. 
Using (\ref{trace-integral}) we can give the definitive definition of the set of Liouville densities:

\vspace{0.2cm}
\begin{equation} \label {densities}
\sS(\sH_n)=\left\{\rho\in\cF^2(\sH_n)\left|\right.\rho(p)\geq0 \,\forall p\in\sP(\sH_n)\,,\, \int_{\sP(\sH_n)}\rho d\nu=1\right\},
\end{equation}
\\
\\
If $\sigma$ and $\sigma'$ are density matrices on $\sH_n$ then their difference $\sigma-\sigma'$ is a Hermitian operator with trace zero so Hilbert-Schmidt distance coincides with the $\cL^2$-distance of associate densities up to a multiplicative factor: If $\rho=\cS(\sigma)$ and $\rho'=\cS(\sigma')$ then:

\begin{equation}\label{dist}
{\int_{\sP(\sH_n)} |\rho-\rho'|^2 d\nu}={\frac{1}{n+1}\left(tr[(\sigma-\sigma')^2]+tr(\sigma-\sigma')^2\right)}={\frac{1}{n+1}tr[(\sigma-\sigma')^2]},
\end{equation}

\vspace{0.5cm}
\noindent
by linearity and formula (\ref{trace-integral2}), i.e. $d_{HS}(\sigma,\sigma')=\sqrt{n+1}d_{\cL^2}(\rho,\rho')$ for every pair $\sigma$ and $\sigma'$ of density matrices.
The following proposition establishes a necesssary and sufficient condition such that a Liouville density describes a pure state:
\begin{proposition}
The Liouville density $\rho\in\sS(\sH_n)$ represents a  pure state if and only if:
\beq\label{purenorm}
\p\rho\p_{2}= \int_{\gP(\sH_n)} |\rho(p)|^2d\nu(p)=\frac{2}{n+1}.
\eeq
\end{proposition}
\begin{proof}
A density matrix $\sigma\in\gB(\sH_n)$ is a pure state ($\sigma=|\psi\rangle\langle\psi|$, for $\psi\in\sH_n$ with $\parallel\psi\parallel=1)$ if and only if $tr(\sigma^2)=1$. Let $\sigma$ be a pure state and $\rho=\cS(\sigma)$. Using (\ref{trace-integral2}), we can calculate:
$$\int_{\gP(\sH_n)} |\rho|^2d\nu=\frac{1}{n+1}(tr(\sigma^2)+1),$$
the condition of purity $tr(\sigma^2)=1$ becomes:
$$\int_{\gP(\sH_n)} |\rho|^2d\nu=\frac{2}{n+1}.$$
\end{proof}

\noindent
Even if $\rho$ is a real function, we write the tautological square absolute value $|\rho|^2$ to stress that (\ref{purenorm}) is the $\mL^2$-norm in $\cF^2(\sH_n)$. 

\vspace{0.5cm}
\section{Re-quantization of the classical-like picture}

The aim of this section is finding a way to calculate explicitly an operator from the associated function on the projective space (re-quantization). In other words we construct the inverse $\cS^{-1}$ of the map $\cS$; This notion will be used further to implement the tensor products of operators in the Hamiltonian formulation. Indeed the term \emph{re-quantization} is used with the following meaning: It is a prescription to associate a self-adjoint operator to each classical-like observable and a density matrix to each Liouville density, thus it is the translation from the Hamiltonian formalism to the standard formalism of QM. 
\begin{theorem}\label{representation}
Let $\sH_n$ a finite-dimensional Hilbert space with dimension $n$ larger than 2. Let $\cF^2(\sH_n)$ be the space of square-integrable frame functions on $\sH_n$. If $\rho\in\cF^2(\sH_n)$, then the operator $\sigma\in\gB(\sH_n)$ such that $\rho(p)=tr(\sigma p)$ (i.e. $\sigma=\cS^{-1}(\rho)$) is given by:

\begin{equation}\label{inverse}
\sigma=(n+1)\int_{\sP(\sH_n)}\rho(p) \,\left(p-\frac{1}{n+1}\bI\right)\,d\nu(p)
\end{equation}
\\
where $\bI$ is the identity operator.
\end{theorem}
\begin{proof}

Let $\varphi$ be a vector of the unit sphere $\bS^{2n-1}=\{\psi\in\sH_n\big|\p\psi\p=1\}$. Since a point  $p\in\sP(\sH_n)$ can be represented by a rank-1 orthogonal projector then we can take the standard expectation value $\langle\varphi|p\,\varphi\rangle=tr(p|\varphi\rangle\langle\varphi|)$.
\\
$f_\varphi(p):=tr(p|\varphi\rangle\langle\varphi|)$ is the function $\cS(|\varphi\rangle\langle\varphi|)$ given by the normalized pure state $|\varphi\rangle\langle\varphi|$. Applying  (\ref{trace-integral2}) we can write the follwoing relation:

$$\int_{\sP(\sH_n)}{f_\varphi(p)}\rho(p) d\nu(p)=\frac{1}{n+1}(tr(|\varphi\rangle\langle\varphi|\, \sigma)+tr(\sigma)),$$
Thus:
$$tr\left(|\varphi\rangle\langle\varphi|\,\sigma\right)=\langle\varphi|\sigma\varphi\rangle=(n+1)\int_{\sP(\sH_n)}\rho(p)\langle\varphi|p\varphi\rangle d\nu(p)-tr(\sigma),$$
\\
the second equality is true for every $\varphi\in\bS^{2n-1}$ i.e. for every $\varphi\in\sH_n$ by sesquilinearity. Thus:
$$\sigma=(n+1)\int_{\sP(\sH_n)}\rho(p) \,p\,d\nu(p)-tr(\sigma) \mathbb I,$$
using (\ref{trace-integral}):
$$\sigma=(n+1)\int_{\sP(\sH_n)}\rho(p) \,p\,d\nu(p)-\int_{\sP(\sH_n)}\rho(p)d\nu(p) \mathbb I,$$
\\
that is the statement of the proposition:

$$\sigma=(n+1)\int_{\sP(\sH_n)}\rho(p) \,\left(p-\frac{1}{n+1}\bI\right)\,d\nu(p).$$
\end{proof}
\noindent
In other words, the action of $\cS^{-1}$ on the function $\rho\in\cF^2(\sH)$ is obtained by the smearing of $\rho$ with the operator:
\beq\label{qd}
\gB(\sH_n)\ni\gS(p):=(n+1)p-\bI\qquad p\in\sP(\sH_n),
\eeq
\beq\label{smearing}
\sigma=\int_{\sP(\sH_n)} \rho(p)\,\gS(p) d\nu(p).
\eeq
\\
For this reason let us call the operator-valued function $\gS:\sP(\sH_n)\rightarrow \gB(\sH_n)$  \emph{re-quantization distribution} since its smearing action on each Liouville density gives the correspondent density matrix. 
The statement of theorem \ref{representation} can be used to construct a re-quantization prescription to obtain a quantum observable (a self-adjoint operator) smearing  a classical-like observable with a re-quantization distribution. We calculate the inverse map of $\cO:iu(n)\ni A\mapsto f_A$ defined in (\ref{o}) in the general form. Let $A\in iu(n)$, by direct computation:
 
$$\int_{\sP(\sH_n)}f_A(p)\gS(p)d\nu(p)=\kappa\int_{\sP(\sH_n)}tr(Ap)\gS(p)d\nu(p)+\frac{1-\kappa}{n} tr(A) \int_{\sP(\sH_n)} \gS(p) d\nu(p),$$
\\
exploiting the statement of theorem \ref{representation} and noting that $\int \gS(p) d\nu(p)=\bI$, we can write:

$$ \int_{\sP(\sH_n)}f_A(p)\gS(p)d\nu(p)=\kappa A+\frac{1-\kappa}{n}tr(A)\bI,$$
\\
an easy computation shows $\int f_A(p) d\nu(p)=tr(A)$ for every $\kappa >0$, thus:

\beq
A=\frac{1}{\kappa}\int_{\sP(\sH_n)}f_A(p)\left[\gS(p)-\frac{1-\kappa}{n}\bI\right]d\nu(p).
\eeq
\\
The general re-quantization distribution for observables, i.e. the operator-valued function $\gO:\sP(\sH_n)\rightarrow \gB(\sH_n)$ such that for any $A\in iu(n)$:

$$A=\int_{\sP(\sH_n)} f_A(p) \gO(p) d\nu(p)$$
is given by:
\beq
\gO(p) = \frac{(n+1)}{\kappa}p-\left(\frac{n+1-\kappa}{\kappa\,n}\right )\bI.
\eeq
\\
In the choice $\kappa=n+1$, where the action of $\cO$ is given by (\ref{osst}), re-quantization distribution is simply $\gO(p)=p$. Summarizing: If the prescription to obtain a classical-like Hamiltonian formulation of a finite-dimensional quantum theory is given by (\ref{osst}) then the re-quantization procedure is given by the following formulas:
\beq
A=\int_{\sP(\sH_n)}f_A(p)p\,d\nu(p),
\eeq
\beq
\sigma=\int_{\sP(\sH_n)} \rho(p)\,\gS(p) d\nu(p),
\eeq
\\
for every classical-like observable (real functions in $\cF^2(\sH_n)$) and Liouville density on $\sP(\sH_n)$.

\vspace{0.5cm}

\section{Composite systems}

In Classical Mechanics the phase space of a composite system is given by the cartesian product of phase spaces of each subsystem. While if one consider a quantum composite system then the phase space must be the projective space of the tensor product of the Hilbert spaces of the subsystems, according to standard Quantum Mechanics. We can consider a bipartite quantum system which consists in two subsystems described in the finite-dimensional Hilbert spaces $\sH$ and $\sK$: The phase space (in the geometric Hamiltonian sense) of the system is given by $\sP(\sH\otimes\sK)$ and not by $\sP(\sH)\times\sP(\sK)$, however the second one is embedded in the first one by \emph{Segre embedding}. 
Let us recall few fundamental ideas: Consider $A\in\gB(\sH)$ and $B\in\gB(\sK)$, the \emph{tensor product of two operators}, $A\otimes B$, can be defined in the following way on the product vectors $\psi\otimes\phi$:

\begin{equation}
 A\otimes B(\psi\otimes\phi):=A\psi\otimes B\phi,
\end{equation}
\\
and the action extends to whole Hilbert space $\sH\otimes\sK$ by linearity. The span of all $A\otimes B$ can be denoted by $\gB(\sH)\otimes\gB(\sK)$ and it coincides with $\gB(\sH\otimes\sK)$ as a general result on Von Neumann algebras. Of course, not all the operators in  $\gB(\sH\otimes\sK)$ are in the product form $A\otimes B$, but considering a general operator in $\gB(\sH\otimes\sK)$ we can define a notion of \emph{restriction} of such operator to $\sH$ or $\sK$, via the so-called \emph{partial trace}.

\begin{definition}\label{pt}
Let be $A\in\gB(\sH\otimes\sK)$. The {\bf partial trace} of $A$ w.r.t. $\sK$ (similarly $\sH$) is the unique operator $tr_\sK (A)\in\gB(\sH)$ such that:

\begin{equation}\label{ptd}
tr\left[tr_\sK (A)B\right]=tr[A( B\otimes \mathbb I_\sK)]\qquad \forall B\in\gB(\sH).
\end{equation}
\\
where $\mathbb I_\sK$ denotes the identity operator on $\sK$.
\end{definition}

Consider a quantum system made up by two quantum subsystems which are described by the observable algebras $\gB(\sH)$ and  $\gB(\sK)$. According to standard quantum theory the observable algebra  of the composite system is given by the tensor product $\gB(\sH)\otimes\gB(\sK)=\gB(\sH\otimes\sK)$.

\begin{definition}
A state $\sigma\in \gB(\sH\otimes\sK)$ is called {\bf separable} if it can be written as:

$$\sigma=\sum_i \lambda_i \sigma^{(1)}_i\otimes \sigma^{(2)}_i,$$
\\
whit weights $\lambda_i> 0$ and states $\sigma_i^{(1)}\in\gB(\sH)$, $\sigma_i^{(2)}\in\gB(\sK);$ otherwise it is called {\bf entangled}. 
\end{definition}
\noindent According to the above definition, a pure state $\sigma\in \gB(\sH\otimes\sK)$  is separable if and only if it is of product form $\sigma=\sigma_1\otimes\sigma_2$. 
In this regard let us introduce the well-known \emph{Segre embedding}. The tensor product map $\otimes:\sH\times\sK\rightarrow\sH\otimes\sK$ induces a canonical embedding of the cartesian product of complex projective spaces in the complex projective space of the tensor product, called {\bf Segre embedding}:

$$Seg:\sP(\sH)\times\sP(\sK)\rightarrow\sP(\sH\otimes\sK),\qquad\qquad$$
\begin{equation}
Seg:(|\psi\rangle\langle\psi,|\phi\rangle\langle\phi|)\mapsto |\psi\otimes\phi\rangle\langle\psi\otimes\phi|,\qquad\p\psi\p=\p\phi\p=1.
\end{equation}

\vspace{0.5cm}
\noindent
The action of Segre embedding can be written as $Seg(p_1,p_2)=p_1\otimes p_2$ for $p_1\in\sP(\sH)$ and $p_2\in\sP(\sK)$, representing pure states as points of projective space. In the standard formulation, the image $Seg(\sP(\sH)\times\sP(\sK))$ gives the separable pure states of the composite system. Here we use the Segre embedding to explicitly construct the isomorphism  between $\cF^2(\sH)\otimes\cF^2(\sK)$ and $\cF^2(\sH\otimes\sK)$.

\begin{proposition}\label{isoseg}
Let $\sH$ and $\sK$ be finite-dimensional Hilbert spaces with $\dim\sH,\dim\sK>2$.
The map $\cI:\cF^2(\sH\otimes\sK)\rightarrow\cF^2(\sH)\otimes\cF^2(\sK)$ defined as the pull-back by Segre embedding:
\begin{equation}\label{iso}
\cI(f)=Seg^* f
\end{equation}
is an isomorphism.
\end{proposition}
\begin{proof}
For any $f\in\cF^2(\sH\otimes\sK)$, its image function $\cI(f):(p_1,p_2)\mapsto f\circ Seg(p_1,p_2)$ belongs to $\cF^2(\sH)\otimes\cF^2(\sK)$. We have to show that $\cI$ is bijective.

The generic element $g$ of $\cF^2(\sH)\otimes\cF^2(\sK)$ is the function given by the finite sum:
\\
$$g:(p_1,p_2)\mapsto\sum_{i\in I} g_1^{(i)}(p_1)g_2^{(i)}(p_2),$$ 
with $g_1^{(i)}\in\cF^2(\sH)$ and  $g_2^{(i)}\in\cF^2(\sK)$ for every $i\in I$. The function $g$ can be written as: 
$$g:(p_1,p_2)\mapsto\sum_{i\in I} tr(A_1^{(i)}p_1)tr(A_2^{(i)}p_2),$$
with $A_1^{(i)}\in\gB(\sH)$ and $A_2^{(i)}\in\gB(\sK)$ for every $i\in I$. We define the action of the map $\cJ: \cF^2(\sH)\otimes\cF^2(\sK)\rightarrow \cF^2(\sH\otimes\sK)$ as: 
$$\cJ(g):\sP(\sH\otimes\sK)\rightarrow \bC$$
$$\cJ(g):p\mapsto\sum_{i\in I} tr\left(A_1^{(i)}\otimes A_2^{(i)}p\right)=tr\left(\sum_{i\in I}A_1^{(i)}\otimes A_2^{(i)}p\right).$$
\\
The direct calculation shows that $\cJ=\cI^{-1}$, so $\cI$ is a bijection. 
\end{proof}
\vspace{0.5cm}
\noindent
One can prove the above result establishes a $C^*$-algebraic isomorphism, however only the isomorphism of vector spaces is useful for us, considering the convex set $\sS(\sH\otimes\sK)\subset\cF^2(\sH\otimes\sK)$ of Liouville denisties. The hypothesis $\dim\sH,\dim\sK>2$ is mendatory becuase we exploit the isomorphisms $\cF^2(\sH)\simeq\gB(\sH)$ and $\cF^2(\sK)\simeq\gB(\sK)$.

\vspace{0.3cm}

\section{Entanglement in the geometric Hamiltonian picture}

In this section we introduce the machinery to describe quantum esntanglement of a bipartite system in terms of Liouville densities defined on the phase space given by the projective Hilbert space. As \emph{inverse-quantization scheme} for states (to obtain Liouville densities from density matrices) we consider the isomorphism of vector spaces $\cS:\gB(\sH\otimes\sK)\rightarrow \cF^2(\sH\otimes\sK)$ given by $\cS(\sigma)(p)=tr(\sigma p)$ for every $\sigma\in\gB(\sH\otimes\sK)$, we also consider the isomorphisms $\cS_\sH:\gB(\sH)\rightarrow \cF^2(\sH)$ and $\cS_\sK:\gB(\sK)\rightarrow \cF^2(\sK)$ defined for the subsystems.

Since a pure state of a bipartite system is separable if and only if it is represented by a pure tensor in $\sH\otimes\sK$, we want to investigate how product form is encoded in frame functions formalism. Henceforth we assume $\dim\sH,\dim\sK>2$ without further specifications.

The following result shows a necessary and sufficient condition on $\rho\in\cF^2(\sH\otimes\sK)$ so that $\rho=\cS(\sigma_1\otimes\sigma_2)$  with $\sigma_1\in\gB(\sH)$ and $\sigma_2\in\gB(\sK)$. In other words there is a criterion to check if a function in $\cF^2(\sH\otimes\sK)$ is associated to an operator in the product form.

\begin{proposition}\label {product form}
Let be $\rho\in\cF^2(\sH\otimes\sK)$, the operator $\sigma=\cS^{-1}(\rho)\in\gB(\sH\otimes\sK)$ is given by a product $\sigma=\sigma_1\otimes \sigma_2$ with  $\sigma_1\in\gB(\sH)$ and $\sigma_2\in\gB(\sK)$ if and only if there are $\rho_1\in\cF^2(\sH)$, $\rho_2\in\cF^2(\sK)$ such that:

\begin{equation}\label{frametensor}
(\rho\circ Seg) (p_1,p_2)=\rho_1(p_1)\rho_2(p_2)\qquad \forall (p_1,p_2)\in\sP(\sH)\times\sP(\sK), 
\end{equation}       
\\
where $Seg:\sP(\sH)\times\sP(\sK)\rightarrow \sP(\sH\otimes\sK)$ is the Segre embedding $Seg: (p_1,p_2)\mapsto p_1\otimes p_2$. Moreover the functions $\rho_1$ and $\rho_2$ satisfy $\rho_1=\cS_\sH(\sigma_1)$ and $\rho_2=\cS_\sK(\sigma_2)$. \\In this case we say that the function $\rho$ is of the {\bf product form} writing $\rho=\rho_1\diamond \rho_2$.    
\end{proposition}

\begin{proof}
Let us suppose that $\rho=\cS(\sigma)$, i.e. $\rho(p)=tr(\sigma p)$, where $\sigma=\sigma_1\otimes \sigma_2$ with $\sigma_1\in\gB(\sH)$ and $\sigma_2\in\gB(\sK)$. Just calculate, for any $p_1\in\sP(\sH)$ and $p_2\in\sP(\sK)$: 

$$(\rho\circ Seg)(p_1,p_2)=\rho(p_1\otimes p_2)=tr(\sigma p_1\otimes p_2)=tr(\sigma_1p_1\otimes \sigma_2p_2)=tr(\sigma_1p_1)tr(\sigma_2p_2) $$
\\
put: $\rho_1(p_1)=tr(\sigma_1p_1)$ and $\rho_2(p_2)=tr(\sigma_2p_2)$.
Thus we proved that (\ref{frametensor}) holds \emph{if}:  

\begin{equation}\label{product}
\rho(p)=tr(\sigma_1\otimes \sigma_2 p)\qquad p\in\sP(\sH\otimes\sK),
\end{equation}
\\
now let us prove that (\ref{frametensor}) \emph{only if} (\ref{product}). The function $\hat\rho:(p_1,p_2)\mapsto\rho_1(p_1)\rho_2(p_2)$ is an element of $\cF^2(\sH)\otimes\cF^2(\sK)$.
Since  $\cF^2(\sH)\otimes\cF^2(\sK)$ and $\cF^2(\sH\otimes\sK)$ are isomorphic (Proposition \ref{isoseg}) then for any function $\hat\rho\in\cF^2(\sH)\otimes\cF^2(\sK)$ there is a unique function 
$\rho\in\cF^2(\sH\otimes\sK)$ such that $\hat\rho=\cI(\rho)=\rho\circ Seg$. Thus  the function $\rho\in\cF^2(\sH\otimes\sK)$ satisfying (\ref{frametensor}) is unique and given by (\ref{product}). 
%Requiring the product form on the image of Segre embedding we have:
%where  $\sigma_1=\cS_\sH^{-1}(\rho_1)$ and $\sigma_2=\cS_\sK^{-1}(\rho_2)$.
\end{proof}
\vspace{0.3cm}
\noindent In proposition \ref{product form}, we have introduced the product $\diamond$ corresponding to tensor product between operators, i.e. $\cS(A\otimes B)=\cS_\sH(A)\diamond\cS_\sK(B)$ for every $A\in\gB(\sH)$ and $B\in\gB(\sK)$.
Since $\cS$ is linear, the vector space $\cF^2(\sH\otimes\sK)$ is the span of all $\rho_1\diamond \rho_2$. Applying the result of  Proposiiton \ref{representation} we can give an explicit definition of the $\diamond$-product.  The natural idea is representing $\sigma_1\otimes\sigma_2$ in terms of the integral introduced in Proposition \ref{representation}. Consider the re-quantization distributions $\gS_H:\sP(\sH)\rightarrow \gB(\sH)$ and $\gS_K:\sP(\sK)\rightarrow\gB(\sK)$ according to definition (\ref{qd}):
\begin{proposition}
Let be $\rho_1\in\cF^2(\sH)$ and $\rho_2\in\cF^2(\sK)$. The function $\rho\in\cF^2(\sH\otimes\sK)$ such that $\cS^{-1}(\rho)=\cS^{-1}_\sH(\rho_1)\otimes\cS^{-1}_\sK(\rho_2)$ is given by:

\begin{equation}\label{frametp}
\rho(p)=\int_{\sP(\sH)\times\sP(\sK)} \rho_1(p_1)\rho_2(p_2) tr\left[p\,\, \gS_H(p_1)\otimes
\gS_\sK(p_2)\right]d\nu_\sH(p_1) d\nu_\sK(p_2)=:(\rho_1\diamond\rho_2)(p),
\end{equation}
\\
where $\nu_\sH$ and $\nu_\sK$ are the Liouville measures respectively defined on the manifolds $\sP(\sH)$ and $\sP(\sK)$. 
\end{proposition}
\begin{proof}
\vspace{0.2cm}
The thesis is a direct result of these two steps: representation of the operator $\cS^{-1}_\sH(\rho_1)\otimes\cS^{-1}_\sK(\rho_2)$ with the integral formula (\ref{smearing}) and the calculation of $\rho=\cS(\cS^{-1}_{\sH}(\rho_1)\otimes\cS^{-1}_\sK(\rho_2))$.
\end{proof}
\vspace{0.5cm}
\noindent 

The product function $\rho_1\diamond\rho_2$ is given by a smearing on the cartesian product $\sP(\sH)\times\sP(\sK)$ with a kernel $\gT:\sP(\sH\otimes\sK)\times\sP(\sH)\times\sP(\sK)\rightarrow \bC$ which does not depend on $\rho_1$ and $\rho_2$ but only on the quantization distributions on $\sP(\sH)$ and $\sP(\sK)$ given by $\gT(p,p_1,p_2)=tr[p\,\gS_H(p_1)\otimes \gS_K(p_2)]$, thus:
\
\beq
(\rho_1\diamond\rho_2)(p)=\int_{\sP(\sH)\times\sP(\sK)}\rho_1(p_1)\rho_2(p_2) \gT(p,p_1,p_2)d\nu_\sH(p_1) d\nu_\sK(p_2).
\eeq
\\
With a very compact notation, we can write: $\rho_1\diamond\rho_2=\int \rho_1\otimes\rho_2 \,\gT\, d\nu_\sH\, d\nu_\sK$.

 We can define an anologous notion of partial trace for functions in $\cF^2(\sH\otimes\sK)$ that are not of product form $\rho_1\diamond\rho_2$. In Definition \ref{pt}, partial trace of $\sigma\in\gB(\sH\otimes\sK)$, denoted as $tr_\sK(\sigma)$, is defined as the unique operator such that 

$$tr(tr_\sK(\sigma)A)=tr(\sigma A\otimes \mathbb I_\sK)\qquad \forall A\in\gB(\sH).$$
\\
A slightly alternative definition is the following: the {\bf partial trace} with respect to $\sK$ is the injective map $tr_\sK:\gB(\sH\otimes\sK)\rightarrow \gB(\sH)$ given by:

\begin{equation}
tr_\sK(\sigma\otimes \sigma'):=\sigma\,\,tr(\sigma')\qquad \forall \sigma\in\gB(\sH), \forall \sigma'\in\gB(\sK),
\end{equation}
\\
and extended to whole $\gB(\sH\otimes\sK)$ by linearity. In order to apply the tool of partial trace in the approach of this paper, we define a map, we can say \emph{partial integral}, using the trace-integral formula already introduced in (\ref{trace-integral}).

\begin{definition} 
Let $\cF^2 (\sH\otimes\sK)\ni \rho\mapsto \rho_\sK \in\cF^2(\sH)$ be a map defined on product elements by:
\begin{equation}\label{partint}
(\rho_1\diamond \rho_2)_\sK(p_1):=\rho_1(p_1)\int_{\sP(\sK)} \rho_2(p_2) d\nu_\sK(p_2),
\end{equation}
for any pair $\rho_1\in\cF^2(\sH)$ and $\rho_2\in\cF^2(\sK)$ and extended by linearity to whole $\cF^2 (\sH\otimes\sK)$. We call the map $\rho\mapsto \rho_\sK$ the {\bf partial integral} w.r.t. $\sP(\sK)$.
\end{definition}
In case of quantum states, the partial integral can be interpreted as the integration of a Liouville density describing a state of a composite system w.r.t. to a marginal measure obtaining a \emph{marginal} probability density.  
\begin{proposition}\label{sint}
 If  $\rho\mapsto \rho_\sK$ is the partial integral on $\cF^2 (\sH\otimes\sK)$ w.r.t. $\sP(\sK)$, then:
\\
{\bf (a)} $\rho_\sK$ has this form: 
\begin{equation}\label{pi}
\rho_\sK(p_1)=\int_{\sP(\sK)} (\rho \circ Seg)(p_1,p_2) d\nu_\sK(p_2)\qquad\forall \rho\in\cF^2(\sH\otimes\sK).
\end{equation}
\\
{\bf (b)} The following relation holds:

\begin{equation}\label{parteq}
tr_K\left[\cS^{-1}(\rho) \right]=\cS_\sH^{-1}(\rho_\sK)\qquad \forall\rho\in\cF^2(\sH\otimes\sK).
\end{equation}
\\
Analogous statements for $\rho\mapsto\rho_\sH$. 
\end{proposition}
\begin{proof}
The generic element $\rho\in\cF^2(\sH\otimes\sK)$ can be written as a finite sum:
$$\rho=\sum_{i\in I}\rho_1^{(i)}\diamond\rho_2^{(i)},$$
with $\rho_1^{(i)}\in\cF^2(\sH)$ and  $\rho_2^{(i)}\in\cF^2(\sK)$ for every $i\in I$. Calculating the partial integral as in (\ref{partint}):
$$\rho_\sK=\sum_{i\in I}\rho_1^{(i)}\int_{\sP(\sK)}\rho_2^{(i)}d\nu_\sK.$$
\\
Let us show that above expression is equivalent to (\ref{pi}): By definition of Segre embedding, we have $(\rho\circ Seg)(p_1,p_2)=\rho(p_1\otimes p_2)$, $\forall (p_1,p_2)\in \sP(\sH)\times\sP(\sK)$, in particular:

$$(\rho\circ Seg) (p_1,p_2)=\sum_{i\in I}  \rho_1^{(i)}(p_1)\rho_2^{(i)}(p_2).$$
\\
Integrating w.r.t. $\nu_\sK$:
$$\int_{\sP(\sK)}\rho\circ Seg \,\,d\nu_\sK=\sum_{i\in I}  \rho_1^{(i)}\int_{\sP(\sK)}\rho_2^{(i)}\,\,d\nu_\sK=\rho_\sK.$$
\\
Let us prove the statement (b). By linearity and Proposition \ref{product form}:

$$\cS^{-1}(\rho)=\sum_{i\in I}\,\,\cS_\sH^{-1}(\rho_1^{(i)})\otimes\cS_\sK^{-1}(\rho_2^{(i)}),$$ 
\\
applying the partial trace $tr_\sK$:
$$tr_\sK\left[\cS^{-1}(\rho)\right]=\sum_{i\in I}\,\,\cS_\sH^{-1}(\rho_1^{(i)})\int_{\sP(\sK)}\rho_2^{(i)}\,d\nu_\sK=\cS_\sH^{-1}\left(\sum_{i\in I}\,\,\rho_1^{(i)}\int_{\sP(\sK)}\rho_2^{(i)}\,d\nu_\sK\right)=\cS_\sH^{-1}(\rho_\sK).$$
\end{proof}
\vspace{0.5cm}
\noindent
The statement of this proposition can be used to prove the next result showing how integrals of frame functions over $\sP(\sH\otimes\sK)$  can be computed.

\begin{theorem} \label{teorema}
Let $\sH$ and $\sK$ be finite-dimensional Hilbert spaces with $\dim\sH,\dim\sK>2$. Consider projective spaces $\sP(\sH)$, $\sP(\sK)$, $\sP(\sH\otimes\sK)$, each equipped with the discussed almost complex K\"ahler structure. $\nu_\sH$, $\nu_\sK$ and $\nu$ denotes the respective Liouville measures. $\cF^2(\sH\otimes\sK)$ denotes the vector space of frame functions in $\cL^2(\gP(\sH\otimes\sK), \nu)$.
\\
The following fact holds for any $\rho\in\cF^2(\sH\otimes\sK)$:
\begin{equation}
\int_{\sP(\sH)\times\sP(\sK)} \rho\circ Seg\,\,\, d\nu_\sH d\nu_\sK=\int_{\sP(\sH\otimes\sK)} \rho\,\, d\nu,
\end{equation} 
where $d\nu_\sH d\nu_\sK$ is the standard product measure on $\sP(\sH)\times\sP(\sK)$.
\end{theorem}
\begin{proof}
Let  $\cS:\gB(\sH\otimes\sK)\rightarrow \cF^2(\sH\otimes\sK)$ be the isomorphism defined as $\cS(\sigma)=\rho$ such that $\rho(p)=tr(\sigma p)$ for every $p\in\sP(\sH\otimes\sK)$. Trace integral formula (\ref{trace-integral}) holds:
$$\int_{\sP(\sH\otimes\sK)} \rho\,\, d\nu=tr\left[\cS^{-1}(\rho)\right].$$
Using statement (b) of Proposition \ref{sint}:
\begin{equation}\label{trtr}
tr\left(tr_\sK\left[\cS^{-1}(\rho)\right]\right)=\int_{\sP(\sH)} \rho_\sK\,\,d\nu_\sH
\end{equation}
Since $tr\left(tr_\sK\left[\cS^{-1}(\rho)\right]\right)=tr\left[\cS^{-1}(\rho)\right]$ by definition of partial trace, the theorem is proved by statement (a) of Proposition \ref{sint}.
\end{proof} 
\vspace{0.5cm}
\noindent
Let us recall the set of Liouville denisties is denoted by $\sS(\sH)$ as in (\ref{densities}) and the subset of densities representing pure states is denoted as $\sS_p(\sH)$. In the following there is the definition of separable and entangled states in terms of Liouville densities.

\begin{definition}\label{sep-ent}
Let $\rho\in\sS_p(\sH\otimes\sK)$ be a Liouville density representing a pure state of the composite system described on $\sP(\sH\otimes \sK)$. $\rho$ is said to be a {\bf separable pure state} if there are $\rho_1\in\sS_p(\sH)$ and $\rho_2\in\sS_p(\sK)$ such that $\rho=\rho_1\diamond \rho_2$.
In other words, $\rho$ is said to be a {\bf separable pure state} if:
\begin{equation}\label{separable}
 (\rho\circ Seg)(p_1,p_2)=\rho_\sK(p_1)\rho_\sH(p_2)  \qquad\forall (p_1,p_2)\in\sP(\sH)\times\sP(\sK) 
\end{equation}
where $\rho_\sH$ and $\rho_\sK$ are the partials integrals of $\rho$ w.r.t. $\sP(\sH)$ and $\sP(\sK)$ respectively.
We denote the set of separable pure states as $\sS_p^{sep}(\sH\otimes\sK)$.

The elements of the convex hull $\sS^{sep}(\sH\otimes\sK):=conv[\sS_p^{sep}(\sH\otimes\sK)]$ are called {\bf separable mixed states}. Finally,  the states belonging to $\cE(\sH\otimes\sK):=\sS(\sH\otimes\sK)\setminus\sS^{sep}(\sH\otimes\sK)$ are called {\bf entangled states}.
\end{definition}

Definition \ref{sep-ent} suggests that the measure of the subset in $\sP(\sH)\times\sP(\sK)$ where the equation $\rho\circ Seg=\rho_\sK\rho_\sH$ fails can be considered an entanglement measure of the state $\rho$.   

From the physical point of view this idea of entanglement measure does not take into account the distinguishability of entangled states. Below the proposal of an entanglement measure based on a $\cL^2$-distance. 
\\
Let us introduce a real map $E:\sS(\sH\otimes\sK)\rightarrow \bR$ defined by:
\begin{equation}\label{entmeasure}
E(\rho):=\left(\int_{\sP(\sH)\times\sP(\sK)} |F_\rho|^2d\nu_\sH d\nu_\sK\right)^{\frac{1}{2}}\qquad\forall\rho\in\sS_p(\sH\otimes\sK),
\end{equation}
where 
\begin{equation}\label{F}
F_\rho(p_1,p_2):= (\rho\circ Seg)(p_1,p_2)-\rho_\sK(p_1)\rho_\sH(p_2),
\end{equation}
\\
and the extension of $E$ to $\sS(\sH\otimes\sK)$ is given by the convex roof:

\begin{equation}\label{convexext}
E(\rho):= \inf_{\rho=\sum \lambda_i\rho_i}\sum_i\lambda_i E(\rho_i)\qquad\forall\rho\in\sS(\sH\otimes\sK),   
\end{equation}
\\
where the infimum is taken on all the possible convex combinations of $\rho$ in terms of pure states $\rho_i\in\sS_p(\sH\otimes\sK)$ and the coefficients $\lambda_i$ are the statistical weights of the mixture. 
Since $F_\rho\in\cL^2(\sP(\sH)\times\sP(\sK),d\nu_\sH d\nu_\sK)$ for any $\rho\in\sS_p(\sH\otimes\sK)$ by definition (\ref{F}), $E(\rho)$ is its $\cL^2$-norm.

\noindent
Another natural idea to define an entanglement measure seems to be given by the calculation of the integral of $F_\rho$ itself on $\sP(\sH)\times\sP(\sK)$, however it is always zero. In fact, by definition of $F_\rho$ and theorem \ref{teorema}, we have:
$$\int_{\sP(\sH)\times\sP(\sK)} F_\rho d\nu_\sH d\nu_\sK\ =\int_{\sP(\sH\otimes\sK)}\rho\circ Seg \,d\nu-\int_{\sP(\sH)}\rho_\sK\, d\nu_\sH\int_{\sP(\sK)} \rho_\sH \, d\nu_\sK=0,$$
\\
since $\rho,\rho_\sK,\rho_\sH$ are each normalized to 1 w.r.t. appropriate measures. 
\\
Let us recall the following technical lemma \cite{GCM2011} about the extension of functions from the extremal points to the convex hull, its statement is a convenient tool to check if the map $E$ is a good entanglement measure.

\begin{lemma}\label{lemma}
Let $X$ be the set of extremal points of a convex set $K$ in a finite dimensional vector space. Let $X_0$ be a compact subset of $X$ and $K_0=conv(X_0)$ its convex hull.

For any non-negative continuous function $E:X\rightarrow \bR^+$ which vanishes exactly on $X_0$, its convex extension, defined as in (\ref{convexext}), is convex on $K$ and vanishes exactly on $K_0$. Moreover, if $E$ is invariant under unitary transformations then its convex extension is so.
\end{lemma}

\noindent
In quantum information theory an axiomatic apporach can be adopted to find good candidates of  entanglement measures (e.g. \cite{Keyl}, \cite{Plenio}), for instance requiring that the candidate function assigns to each quantum state of a bipartite system a positive real number and it vanishes on separable states. Another requirement is the invariance of the entanglement measure w.r.t. local unitary transformations. The entanglement measure should be a convex function beacuse entanglement cannot be generated by mixing two states, moreover it should be a continuous function for this physical reason: A small perturbation of a state must correspond to a small change of entanglement.     
The following proposition shows that $E$ satisfies a list of properties of a good entanglement measure.
\begin{proposition}\label{}
The map $E:\sS(\sH\otimes\sK)\ni\rho\mapsto E(\rho)$  satisfies the following properties:
\\
{\bf i)}\,\,\,\,\, $E(\rho)\in\bR^+$ for every  $\rho\in\sS(\sH\otimes\sK)$;
\\
{\bf ii)} \,\,\,$E(\rho)=0$ if and only if $\rho$ is separable;
\\
{\bf iii)} \,$E$ is invariant under the action of the unitary group;
\\
{\bf iv)}\, $E$ is a convex function;
\\
{\bf v)} \,\,$E$ is continuous w.r.t. the uniform norm topology.
\\

\end{proposition}

\begin{proof}
i) $E(\rho)$ is the integral of a non-negative function for any pure state $\rho$. Convex combinations preserve non-negativity.
\\
ii) The non-negative function $|F_\rho|^2$ vanishes everywhere on $\sP(\sH)\times\sP(\sK)$ if and only if $\rho$ is a separable pure state. The proof for mixed state is in iv) below.
\\
iii) The action of the unitary group on $\cF^2(\sH_n)$ is given by $[U(f)](p)=f(UpU^{-1})$, where $U\in U(n)$ and we used the same symbol for the representative operator. We need to prove that:
$$E(U\otimes V \rho)=E(\rho),$$
for every $U\in U(n)$, $V\in U(m)$ where $\dim\sH=n$ and $\dim\sK=m$.
$$E(U\otimes V \rho)=\left(\int |F_{U\otimes V\rho}(p_1,p_2)|^2 d\nu_\sH(p_1) d\nu_\sK (p_2)\right)^{\frac{1}{2}}\qquad\qquad\qquad\qquad\qquad$$
$$= \left(\int |F_{\rho}(Up_1U^{-1},Vp_2V^{-1})|^2 d\nu_\sH(p_1) d\nu_\sK (p_2)\right)^{\frac{1}{2}}\quad$$
$$=\left( \int |F_{\rho}(p_1,p_2)|^2 d\nu_\sH(Up_1U^{-1}) d\nu_\sK (Vp_2V^{-1})\right)^{\frac{1}{2}}\quad $$
$$= \left(\int |F_{\rho}(p_1,p_2)|^2 d\nu_\sH(p_1) d\nu_\sK (p_2) \right)^\frac{1}{2}= E(\rho),\qquad\,\,\,$$
\\
where we used the unitary invariance of the measures $\nu_\sH$ and $\nu_\sK$. The identity  $F_{U\otimes V\rho}(p_1,p_2)=F_{\rho}(Up_1U^{-1},Vp_2V^{-1})$, that is valid for any pair $(p_1,p_2)$, can be checked directly from definition (\ref{F}). The result holds even for mixed states, see lemma \ref{lemma}.
\\
v) Consider a sequnece of pure states $\{\rho_n\}$ that is uniformly convergent to $\rho\in\sS_p(\sH\otimes\sK)$, thus we have the pointwise convergence $\rho_n\rightarrow \rho$ as $n\rightarrow\infty$. Then $\rho_n\circ Seg\rightarrow \rho\circ Seg$ pointwise.
\\
$\{\rho_n\circ Seg\}$ is a sequence of positive bounded functions thus it is dominated by an integrable function and we can apply the dominated convergence theorem, obtaining:
$$\lim_{n\rightarrow \infty} {\rho_n}_\sK=\lim_{n\rightarrow \infty}\int (\rho_n\circ Seg)d\nu_\sK=\int (\rho\circ Seg)d\nu_\sK=\rho_\sK.$$
There is pointwise convergence of the sequences of partial integrals: ${\rho_n}_\sK\rightarrow \rho_\sK$, ${\rho_n}_\sH\rightarrow \rho_\sH$.
Thus we have the following pointwise limit:
$$\lim_{n\rightarrow\infty} F_{\rho_n}(p_1,p_2)=F_\rho(p_1,p_2)\qquad\forall (p_1,p_2)\in\sP(\sH)\times\sP(\sK).$$
Applying the dominated convergence theorem once again:
$$\lim_{n\rightarrow\infty}E(\rho_n)=\lim_{n\rightarrow\infty}\sqrt{\int |F_{\rho_n}|^2d\nu_\sH d\nu_\sK}=\sqrt{\int |F_\rho|^2d\nu_\sH d\nu_\sK}=E(\rho).$$
\\
iv) We apply lemma \ref{lemma}. $E:\sS(\sH\otimes\sK)\rightarrow \bR^+$ is the convex extension of a non-negative continuous function defined on the extremal elements of $\sS(\sH\otimes\sK)$ that vanishes on the separable pure states, then it is a convex function vanishing exactly on the set of separable states.

\end{proof}
\vspace{0.5cm}
\noindent
In standard QM, state distinguishability is quantified by the \emph{trace-distance} between density matrices: $d(\sigma,\sigma')=\frac{1}{2}\p \sigma-\sigma'  \p_1=\frac{1}{2}tr(|\sigma-\sigma'|)$. Thus a good entanglement measure on the set of density matrices should be continuous w.r.t. the topology induced by $\p\,\,\p_1$. 
If $\sH$ is a finite-dimensional Hilbert space then the topology induced by $\p \,\,\p_1$ on $\gB(\sH)$ coincides with the topology induced by the norm $\p T\p:=\sup_{\p\psi\p=1} |\langle \psi |T \psi\rangle|=\sup_{p\in\sP(\sH)}|tr(Tp)|=\p\cS(T)\p_\infty$. For this reason the continuity w.r.t. the uniform norm topology is remarkable in order to use $E$ as an entanglement measure.

\begin{definition}
Let $\sS(\sH\otimes\sK)$ be the set of Liouville densities on $\sP(\sH\otimes\sK)$ describing physical states of a bipartite quantum system. The map $E:\sS(\sH\otimes\sK)\rightarrow \bR$ defined on the extremal points by (\ref{entmeasure}) and extended by (\ref{convexext}) to the convex hull is called {\bf standard Hamiltonian entanglement measure}.
\end{definition}

In the introductory section, we stress that the Hilbert-Schmidt distance between density matrices coincides up to a multiplicative constant with the $\cL^2$-distance between associate Liouville densities. Thus we can express in terms of Liouville densities a well-known entanglement measure defined as the Hilbert-Schmidt distance of an entangled state from the set of separable states. Consider a density matrix $\sigma_*\in\gB(\sH)$ of a bipartite system, an entanglement measure proposed in  \cite{Witte} is:

\begin{equation}\label{HSO}
D(\sigma_*)=\min_{\sigma\in\sS^{sep}}\p\sigma-\sigma_*\p_{HS},
\end{equation}
\\   
where $\sS^{sep}$ is the convex set of separable density matrices. Thus we can introduce another measure of entanglement carried by a Liouville density $\rho\in\sS(\sH\otimes\sK)$ in thi way:

\begin{equation}\label{HSm}
D(\rho)=\min_{\eta\in\sS^{sep}(\sH\otimes\sK)}\sqrt{\int_{\sP(\sH\otimes\sK)} |\rho-\eta|^2 d\nu}.
\end{equation}
\\
Even this definition is based on a $\cL^2$-distance but in the space $\cL^2(\sP(\sH\otimes\sK),d\nu)$ instead in $\cL^2(\sP(\sH)\times\sP(\sK),d\nu_\sH d\nu_\sK)$ like in our proposal. The letter has no a direct analogous in the standard formalism of density matrices, because $\mL^2$-distance is computed for functions that are not defined on the projective space $\sP(\sH\otimes\sK)$ but on the cartesian product $\sP(\sH)\times\sP(\sK)$ which is the \emph{classical-like} phase space of the bipartite system.

The entanglement measure based on Hilbert-Schmidt distance is connected with the violation degree of a generalized Bell inequality as shown by Bertlmann-Narnhofer-Thirring theorem \cite{thirring}. To study this connection from the point of view of Hamiltonian formalism we introduce the \emph{witness inequality} in the next section.

\vspace{0.5cm}

\section{Separability criteria for Liouville densities}
Using the developed machinery we can translate two celebrated separability criteria in the language of Hamiltonian formulation. 

\begin{proposition}\label{sep1}
For any Liouville density $\rho\in\sS(\sH\otimes\sK)$ representing an entangled state there is an observable $f:\sP(\sH\otimes\sK)\rightarrow \bR$ (called \textbf{entanglement witness}) such that:
\begin{equation}\label{WI}
\int_{\sP(\sH\otimes\sK)} f\rho d\nu< 0\qquad\mbox{and}\qquad\int_{\sP(\sH\otimes\sK)} f\eta d\nu\geq 0, 
\end{equation}
\\
for every separable Liouville density $\eta$.
\end{proposition}
\begin{proof}
For any entangled density matrix $\sigma$ in $\sH\otimes\sK$ there is a Hermitian operator $A$ (i.e. a quantum observable) such that $tr(A\sigma)<0$ and $tr(A\theta)\geq0$ for every separable density matrix $\theta$ (see e.g.  Lemma 1 \cite{Horodecki}). Applying trace-integral formulas:

$$tr(A\sigma)=\int_{\sP(\sH\otimes\sK)} f\rho d\nu\quad,\quad tr(A\theta)=\int_{\sP(\sH\otimes\sK)} f\eta d\nu$$
\\
where $\rho=\cS(\sigma)$, $\eta=\cS(\theta)$ and $f=\cO(A)$, i.e. $f$ represents a quantum observable.\\
\end{proof}
\vspace{0.5cm}
\noindent
To make above result useful, we recall when a real function $f\in\cL^2(\sP(\sH\otimes\sK),d\nu)$ represents a quantum observable (i.e. when it verifies $f=\cO(A)$ for some  $A\in iu(n)$). A necessary and sufficient condition, obtained applying proposition 25 in \cite{DV2}, is:

\begin{equation}
\int_{\sP(\sH\otimes\sK)} f \cS(p_0) d\nu =d^2 f(p_0),
\end{equation}
\\
for every pure state $p_0$, where $d=\dim\sH\times\dim\sK$. 

An \emph{entanglement witness} can be defined as a non-positive observable such that its expectation value on every separable state is a positive number. The second inequality in (\ref{WI}) is violated by an entangled state, first equation in (\ref{WI}). The violation of the inequality:  

\begin{equation}\label{Bell}
\int_{\sP(\sH\otimes\sK)} f\rho d\nu\geq 0\qquad \mbox{with $f$ entanglement witness}
\end{equation}       
\\
is a good criterion to test if a state is entangled, it can be called \emph{generalized Bell inequality} in the Hamiltonian formalism.
 The maximal violation of operatorial generalized Bell inequality is connected with Hilbert-Schmidt entanglement measure (\ref{HSO}) by the Bertlmann-Narnhofer-Thirring theorem \cite{thirring}. Thus (\ref{Bell}) and (\ref{HSm}) can be used to obtain a Hamiltonian version of BNT theorem.

\begin{proposition}\label{sep2}
A Liouville density $\rho\in\sS(\sH\otimes\sK)$ describes a separable state if and only if:
\begin{equation}\label{1}
\int_{\sP(\sH\otimes\sK)} \rho f d\nu\geq 0,
\end{equation}
for any quantum observables $f$ satisfying:
\begin{equation}\label{2}
\int_{\sP(\sH\otimes\sK)}\eta_1\diamond\eta_2 f d\nu\geq 0
\end{equation}
 for all $\eta_1\in\cF^2(\sH)$ and 
$\eta_2\in\cF^2(\sK)$ such that:
\beq\label{A}
G(d\eta_1,d\eta_1)=2 (\eta_1-\eta_1^2)
\eeq
\beq\label{B}
R(d\eta_2,d\eta_2)=2 (\eta_2-\eta_2^2)
\eeq
\\
\noindent
where $G$ and $R$ are the scalar products of one-forms respectively induced by the Fubini-Study metrics on $\sP(\sH)$ and $\sP(\sK)$.
\end{proposition}
\begin{proof}
A density matrix $\sigma$ on $\sH\otimes\sK$ is separable if and only if $tr(\sigma A)\geq 0$ for any Hermitian operator $A$ satisfying $tr((P\otimes Q)A)\geq 0$ for all orthogonal projectors $P$ and $Q$ on $\sH$ and $\sK$ respectively (e.g. theorem 1 in \cite{Horodecki}). The statement of the proposition is the translation of this fact in the classical-like functions formalism. Consider the real functions $g=\cS(\tau)\in\cF^2(\sH)$ and $g'=\cS(\tau')\in\cF^2(\sH)$ for $\tau,\tau'\in\gB(\sH)$ , then a direct computation \cite{DV2} produces:

$$\cS(\tau\tau')=\cS(\cS^{-1}(g)\cS^{-1}(g'))=\frac{i}{2}\{g,g'\}_{PB}+\frac{1}{2}G(dg,dg')+gg'.$$
\\
If $\tau'=\tau$ then the Poisson bracket is zero, moreover if $\tau$ is an orthogonal projector:

\beq\label{idem}
g=\cS(\tau)=\frac{1}{2}G(dg,dg)+g^2,
\eeq
\\
i.e. $G(dg,dg)=2(g-g^2)$. The converse is true because $\cS$ is bijective: If $g=\cS(\tau)$ satisfies (\ref{idem}) then $\tau$ is an orthogonal projector. 

If $\eta_1\in\cF^2(\sH)$ satisfies (\ref{A})  and $\eta_2\in\cF^2(\sK)$ satisfies (\ref{B}) then the operators $P=\cS_\sH^{-1}(\eta_1)$ and $Q=\cS_\sK^{-1}(\eta_2)$ are orthogonal projectors. And using (\ref{frametp}) as the definition of $\diamond$ we have $\cS(P\otimes Q)=\eta_1\diamond\eta_2$. We can use trace-integral formulas as in the proof of proposition \ref{sep1} to obtain (\ref{1}) and (\ref{2}). 
\end{proof}
\vspace{0.5cm}

\section{Conclusions and perspectives}

Finite-dimensional quantum systems can be described as Hamiltonian systems in a classical-like fashion via an \emph{inverse quantization} prescription associating functions on the projective space $\sP(\sH)$ to operators in $\gB(\sH)$ in order to give a classical-like representation of quantum observables and quantum states which can be \emph{completely} characterized in these terms \cite{DV2} (exploiting the machinery of \emph{frame functions}). This paper discussed a \emph{re-quantization} of the Hamiltonian formulation to obtain a quantum theory in the standard formalism (self-adjoint operators and density matrices on Hilbert space). Then we have used some developed machinery to apply the Hamiltonian formulation to the description of composite quantum systems and study some basic issues of relevance in quantum information theory: In particular introducing the notions of quantum entanglement and entanglement measure. We have explored the connection between the phase space of a composite system $\sP(\sH\otimes\sK)$ and its classical counterpart $\sP(\sH)\times\sP(\sK)$ that is embedded in the first one by Segre embedding, the entanglement measure is constructed in terms of an integration over $\sP(\sH)\times\sP(\sK)$. Moreover an analogous of the \emph{partial trace} is defined and its interpretation as a marginal probability measure is pointed out. The so-called \emph{partial integral} has been used to define the standard Hamiltonian entanglement measure discussed as a natural idea to quantify entanglement of a quantum state described by a Liouville density on $\sP(\sH\otimes\sK)$. Several properties, generally required to an entanglement measure are checked considering the proposed standard Hamiltonian entanglement measure. 

The analogous of Hilbert-Schmidt entanglement measure has been introduced in terms a of $\mL^2$-distance from the set of separable states and its connection with \emph{generalized Bell inequality} is discussed in terms of Hamiltonian formalism. Finally, propositions \ref{sep1} and \ref{sep2} establish two separability criteria, more precisely two known separability criteria for density matrices have been characterized in terms of Liouville densities.

From the physical point of view, an urgent open issue is probably the determination of finite-dimensional composite quantum systems for which the presented Hamiltonian approach is definetly more convenient than the standard QM.
A subject for further works could be the definition of \emph{completely positive maps} within Hamiltonian framework and their characterization in terms of  equivalent versions of Stinespring's decomposition and Krauss operators in order to state a theory for open systems. Moreover the notion of \emph{quantum operations} can lead to a pure geometric Hamiltonian formulation of basic tasks of quantum computing.

\end{document}